\documentclass[letterpaper, 10 pt, conference]{ieeeconf}  

\IEEEoverridecommandlockouts                              
\usepackage{amsmath,amssymb,amsfonts,bm, mathtools, dsfont}
\usepackage{float}
\usepackage{color}
\usepackage{graphics} 
\usepackage{epsfig} 
\usepackage{mathptmx} 
\usepackage{times} 
\usepackage{amsmath} 
\usepackage{amssymb}  
\usepackage{textcomp}

\usepackage{amsthm}

\usepackage{booktabs}
\usepackage{float}
\usepackage{multirow}
\usepackage{tikz}
\usepackage{cleveref}
\usepackage{cite}
\usepackage{amsmath}
\newtheoremstyle{boldtheorem}
  {3pt}   
  {3pt}   
  {}      
  {}      
  {\bfseries}  
  {.}     
  {1em}   
  {}      

\newtheorem{theorem}{Theorem}
\newtheorem{lemma}{Lemma}
\newtheorem{definition}{Definition}
\allowdisplaybreaks[1]

\title{\LARGE \bf
Data-Driven Reachability with Scenario Optimization \\
and the Holdout Method
}

\author{Elizabeth Dietrich, Rosalyn Devonport, Stephen Tu, and Murat Arcak 
\thanks{E. Dietrich and M. Arcak are with the University of California, Berkeley, USA. Email: {\tt\small \{eadietri, arcak\}@berkeley.edu}. R. Devonport is with the University of New Mexico, Alburquerque, USA. Email: {\tt\small devonport@unm.edu}. S. Tu is with the University of Southern California, USA. Email: {\tt\small stephen.tu@usc.edu}.}
}

\usepackage{tikz} 
\newcommand\copyrighttext{%
	\footnotesize \copyright 2025 IEEE. Personal use of this material is permitted. Permission from IEEE must be obtained for all other uses, in any current or future media, including reprinting/republishing this material for advertising or promotional purposes, creating new collective works, for resale or redistribution to servers or lists, or reuse of any copyrighted component of this work in other works.}
\newcommand\copyrightnotice{%
	\begin{tikzpicture}[remember picture,overlay]
		\node[anchor=south,yshift=10pt] at (current page.south) {\fbox{\parbox{\dimexpr\textwidth-\fboxsep-\fboxrule\relax}{\copyrighttext}}};
	\end{tikzpicture}%
}

\begin{document}

\maketitle
\thispagestyle{empty}
\pagestyle{empty}


\begin{abstract}
Reachability analysis is an important method in providing safety guarantees for systems with unknown or uncertain dynamics. Due to the computational intractability of exact reachability analysis for general nonlinear, high-dimensional systems, recent work has focused on the use of probabilistic methods for computing approximate reachable sets. In this work, we advocate for the use of a general purpose, practical, and sharp method for data-driven reachability: \emph{the holdout method}. Despite the simplicity of the holdout method, we show---on several numerical examples including scenario-based reach tubes---that the resulting probabilistic bounds are substantially sharper and require fewer samples than existing methods for data-driven reachability. Furthermore, we complement our work with a discussion on the necessity of probabilistic reachability bounds. We argue that any method that attempts to \emph{de-randomize} the bounds, by converting the guarantees to hold deterministically, requires (a) an exponential in state-dimension amount of samples to achieve non-vacuous guarantees, and (b) extra assumptions on the dynamics.
\end{abstract}

\copyrightnotice

\vspace*{-2mm}
\section{Introduction}
Reachability analysis plays an integral role in analyzing system safety by determining the set of states a system can transition to in a finite time horizon \cite{althoff2010reachability}. Many set-based techniques \cite{10.1007/978-3-540-24743-2_32, annurev-control-060117-104941, annurev-control-071420-081941, 10.1007/978-3-540-31954-2_19} have been developed to compute under- or over-approximated reachable sets to capture all possible trajectories and present guarantees on state-space safety. However, when there is limited information on a system's dynamics, through only simulation or experimentation, we must estimate reachable sets and derive probabilistic guarantees of correctness directly from data. To provide probability measures for learned reachable sets, many approaches require a minimum number of samples \cite{pmlr-v120-devonport20a, 9147918, 8814354, 10138798} or assumptions on the system dynamics \cite{Lew2020SamplingbasedRA, 10383270, 10068731}, and are often computationally intensive \cite{10.1007/978-3-030-99524-9_17, 10.1109/CDC45484.2021.9682860}. We advocate for the use of a general purpose method, \textit{the holdout method}, that circumvents these limitations and improves the sample complexity and sharpness of existing probabilistic reachability bounds. 

The holdout method, or cross-validation, is one of the simplest and most widely used methods for estimating prediction error and providing statistical guarantees \cite{stone1974cross, Kohavi1995, hastie2009elements, tempobook}. While this method has been utilized across disciplines for decades, its biggest limitation is the availability of data. The holdout method requires a validation set that is only used to assess the performance of the prediction model. If data is scarce, this can lead to test-set contamination or a small test-set size. However, in the context of data-driven reachability analysis, especially simulation-based approaches, generating new data sets avoids these issues and can be computationally efficient.  

There are numerous works investigating other statistical learning techniques for data-driven reachability analysis, including approaches such as scenario optimization \cite{pmlr-v242-dietrich24a, pmlr-v242-lin24a, 8882241}, conformal prediction \cite{10384213, pmlr-v204-tebjou23a, 10383723}, and Gaussian processes \cite{10383270, 7039601}. While many of these works rely on a-priori sample complexities, \cite{pmlr-v242-dietrich24a} calculates a-posteriori probability bounds using a wait-and-judge perspective. Therefore, we compare the a-posteriori bounds we achieve with the holdout method directly to the wait-and-judge perspective to demonstrate the stronger and more computationally-efficient guarantees that we are able to obtain. Further, in contrast to existing approaches, the holdout method provides a general purpose framework, amenable to any reachable set estimator, and it does not require any prior knowledge about system dynamics or sample complexity. 

To conclude our work, we include a discussion on the necessity of probabilistic bounds 
in the context of data-driven reachability analysis. We argue that any attempt to \emph{de-randomize} a probabilistic bound---that is, convert a probabilistic bound into a deterministic bound via enlargement of the reachable set---must not only make extra assumptions on the dynamics (e.g., Lipschitz bounds on the transition function), but also requires an exponential in state-dimension number of scenarios to obtain non-vacuous guarantees.


\section{Problem Statement}
\label{prelim}

\subsection{Forward Reachable Sets.} A forward reachable set is defined as $R = \{\Phi(t_1;t_0, x_0, d) : x_0 \in X_0, d \in D\}$ where $X_0 \subseteq \mathbb{R}^{n_x}$ is the set of initial states, $D$ is the set of disturbance signals $d: [t_0, t_1] \rightarrow \mathbb{R}^{n_d}$, and $\Phi : X_0 \times D \rightarrow \mathbb{R}^{n_x}$ is the state transition function. This is the set of all states to which the system can transition to at time $t_1$ from states $X_0$ at time $t_0$ subject to disturbances in $D$. Since we cannot compute exact reachable sets, we aim to compute an approximation, $\hat{R}$, that is close to the true reachable set in a probabilistic sense.

To compute such an approximation, we first endow both $X_0$ and $D$ with probability distributions $\mu_{X_0}$ and $\mu_D$, respectively. Let $\Delta$ be the resulting probability space from which we draw samples $\delta^{(i)} =\Phi(t_1;t_0, x_{0i}, d_i), i = 1, \dotsc, N $ where $x_{01}, \dotsc, x_{0N} \overset{i.i.d}{\sim} \mu_{X_0}$, $d_1, \dotsc, d_N \overset{i.i.d}{\sim} \mu_{D}$. We compute a reachable set estimate in the form of a sublevel set of a parameterized function 
\begin{equation}
\label{eq:reachability}
\begin{aligned}
\hat{R}(\theta) = \{x \in \mathbb{R}^{n_x} : g(x, \theta) \leq 0\}
\end{aligned}
\end{equation} 
where $g : \mathbb{R}^{n_x} \times \mathbb{R}^{n_{\theta}} \rightarrow \mathbb{R}$. In~\eqref{eq:reachability}, $\theta$ represents a parameterization of the class of admissible reachable set estimators: to fix a value of $\theta$ is to choose an estimator. 

\subsection{Violation Probability.}  
Given the samples $\delta^{(1)}, \dots, \delta^{(N)}$ and a desired confidence parameter $\beta$, we wish to find the minimum-volume reachable set that contains the samples and satisfies the probabilistic guarantee $\mathbf{P} \{ V(\hat{R}(\theta)) > \epsilon\} \leq \beta$. The violation probability $V(\hat{R}(\theta))$, or true error $e$, of the reachable set estimate, $\hat{R}(\theta)$, is defined as the probability that an unseen scenario will violate the reachable set: 
    \begin{equation}
        e \equiv V(\hat{R}(\theta)) \equiv  \mathbf{P}_{\delta \sim \Delta}\{ g(\delta, \theta) > 0\} \equiv \mathbf{P}_{\delta \sim \Delta}\{ \delta \not\in \hat{R}(\theta) \}.
        \label{eq:true}
    \end{equation}

In other words,  if $V(\hat{R}(\theta)) \leq \epsilon$, then our reachable set estimate is robust against constraint violation at level $\epsilon$. Since the true error is not an observable quantity, we utilize the empirical error. Given a new sample set of size $M$ drawn from $\Delta$, the empirical error $\hat{e}$, or test error, is the observed number of scenarios that violate the reachable set estimate:
\begin{align}
    \hat{e} \equiv \hat{V}(\hat{R}(\theta)) &\equiv  \mathbf{P}_{\delta_s \sim \Delta}\{ \delta_s \not\in \hat{R}(\theta) \} \\
    &\equiv \frac{1}{M}\sum^{M}_{i=1 }\mathds{1}_{\hat{R}(\theta)}(\delta_s^{(i)})
    \label{eq:empirical}
\end{align}
where
\begin{equation}
    \mathds{1}_{\hat{R}}(\theta)(\delta_s^{(i)}) = \Bigg\{ \begin{array}{ll}
      1 & \text{if } \delta_s^{(i)} \notin  \hat{R}(\theta)\\
      0 & \text{else}. \\
\end{array} 
\label{classifier}
\end{equation}
Given 
$\hat{e}$ and $\beta$, we will obtain an appropriate $\epsilon$ to satisfy the bound $\mathbf{P} \{ V(\hat{R}(\theta)) > \epsilon\} \leq \beta$ in Section \ref{holdout}.

\subsection{Nonconvex Scenario Reachability Analysis}
To calculate reachable set estimates of the form \eqref{eq:reachability}, we utilize scenario optimization. Scenario optimization is an approach to solving chance-constrained optimization problems by solving a non-probabilistic relaxation of the original problem \cite{dembo1991scenario}. We fix a functional ${\rm Vol}: \mathbb{R}^{n_{\theta}} \rightarrow \mathbb{R}_{\geq 0}$ that acts as a proxy for the volume of $R(\theta)$. This motivates the following scenario program:
\begin{equation}
\label{eq:volprox}
\begin{aligned}
& \underset{\theta}{\text{minimize}}
& & \rm{Vol}(\theta) \\
& \text{subject to}
& & g(\delta^{(i)}, \theta) \leq 0, i = 1, \dotsc, N\\
& & & \theta \in \mathbb{R}^{n_{\theta}}.
\end{aligned}
\end{equation} 

The solution to (\ref{eq:volprox}) is the minimum-volume set that contains sample points $\delta^{(1)}, \dotsc, \delta^{(N)}$. In \cite{pmlr-v242-dietrich24a}, we present two methods for estimating these nonconvex reachable sets, including the sum of radial basis functions, as demonstrated in Section \ref{applications}.

\subsection{Wait-and-Judge}
Given a scenario program of the form \eqref{eq:volprox}, one approach for computing probabilistic bounds is \textit{wait-and-judge} \cite{pmlr-v242-dietrich24a}. The wait-and-judge perspective, combined with nonconvex scenario optimization \cite{8299432},
determines $\epsilon$ a-posteriori as a function of support scenarios. A support scenario is any scenario whose removal changes the solution to \eqref{eq:volprox}. This approach requires re-solving the scenario program upon removal of each individual scenario. Therefore, this approach can be computationally intensive for systems with high dimensionality or complexity. We refer readers to \cite{pmlr-v242-dietrich24a} for more details on computing $\epsilon$. Since we calculate the probability measure of our reachable set after solving the optimization problem, the nonconvex scenario approach is amenable to alternative statistical analysis techniques, such as the holdout method, as discussed in Section \ref{holdout}.

\subsection{Binomial Tail Inversion.}
\label{bti}
To calculate probabilistic bounds for \eqref{eq:volprox}, we employ the holdout method given a sample set of size $M$. We would like to calculate the probability of observing at most $k$ reachable set violations out of $M$ samples. From a statistical perspective, this can be viewed as computing a tail bound for a binomial distribution. We use a binomial tail inversion to obtain the largest true error $e$ such that the probability of $k$ or more samples violating the reachable set is at least $\beta$ \cite{langford2005tutorial}.
This calculates the bound on the true error given the empirical error $\hat{e}$, or the empirical count of boundary violations $\hat{k}$, and confidence $\beta$.
\begin{definition}
\label{def:binomial_tail_inversion}
(Binomial Tail Inversion)
For $\hat{k}$ violations out of $M$ newly sampled scenarios and all $\beta \in (0, 1]$:
    \begin{equation}
        \overline{\text{Bin}}(\hat{k}, M, \beta) = \max_{e}\Bigl\{ e : \text{Bin}\Bigl(\hat{k}, M, e \Bigr) \geq \beta \Bigr\}, \,\,\text{where}
        \label{binominversion}
    \end{equation}
    \begin{equation}
        \text{Bin}\Bigl(\hat{k}, M, e \Bigr) =  \sum_{j=0}^{\hat{k}} \binom Mj e^j (1-e)^{M-j}.
    \end{equation}
    \label{def1}
\end{definition}

\section{The Holdout Method}
\label{holdout}
We present a technique that utilizes the \textit{holdout method} to a-posteriori evaluate the robustness level of the scenario solution to \eqref{eq:volprox} by obtaining an empirical estimate of the accuracy of a given reachable set approximation. This well-known method offers a modular technique for computing probabilistic bounds. In particular, we demonstrate the use of the holdout method in nonconvex scenario reachability analysis; we tighten the probabilistic bounds that were first introduced in \cite{pmlr-v242-dietrich24a}, as shown in Fig. \ref{overview}.  However, it is important to mention that the holdout method works for \textit{any} parameterization of a reachable set estimator, including those that arise from neural networks.

\begin{figure*}[h]
   \centering
    \resizebox{0.9\linewidth}{!}{
        \input{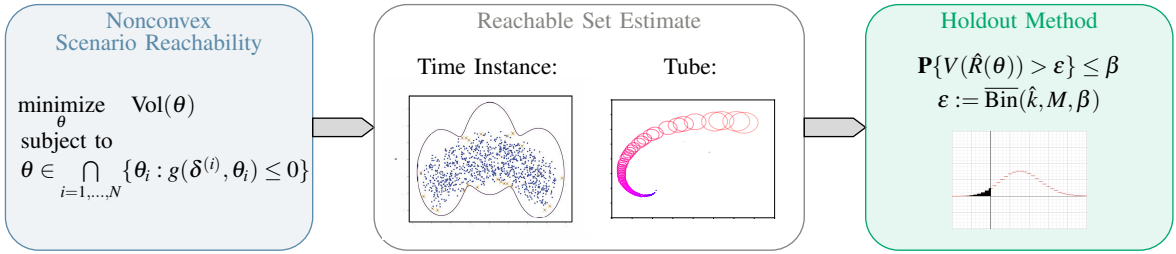}
    } 
   \caption{Using nonconvex scenario reachability analysis, we calculate reachable set estimates, $\hat{R}(\theta)$, in the form of a constrained optimization problem. Specifically, we construct $\hat{R}(\theta)$ from a finite set of radial basis functions. We illustrate the output of this method through reach sets of a single time instance and reach tubes. Finally, we utilize the holdout method, and a binomial tail inversion, to calculate probabilistic bounds for $\hat{R}(\theta)$.}
   \label{overview}
\end{figure*}
The holdout method is a foundational paradigm of data-driven methodologies. It 
employs 
a test set of $M$ fresh scenarios 
to provide a good risk estimate for a given model \cite{hardtrecht2022patterns}. Therefore, we draw a \textit{new} set of samples $\delta^{(i)}_s =\Phi(t_1;t_0, x_{0i}^s, d_i^s), i = 1, \dotsc, M $ where $x_{01}^s, \dotsc, x_{0M}^s \overset{i.i.d}{\sim} \mu_{X_0}$, $d_1^s, \dotsc, d_M^s \overset{i.i.d}{\sim} \mu_{D}$, and test the accuracy of our estimate $\hat{R}(\theta)$ on $\delta_s^{(1)}, \dotsc, \delta_s^{(M)}$. 
Given $\hat{k}$ reachable set violations out of $M$ samples, or the empirical error \eqref{eq:empirical}, we are able to use a binomial tail bound, as introduced in Section \ref{bti}, to calculate a bound on the true error \eqref{eq:true} of the reachable set. We formulate this in the following theorem:

\begin{theorem} (\textit{Adapted from \cite[Thm. 3.3]{langford2005tutorial}})
\label{thm:holdout}
Given any $\beta \in (0, 1)$, empirical count of boundary violations $\hat{k}$, and size of test set $M$, the following probability bound holds for the reachable set estimate $\hat{R}(\theta)$: 
\begin{equation}
\begin{aligned}
\label{equ:probviolate}
 {\mathbf{P}} \{ V(\hat{R}(\theta)) > \overline{\text{Bin}}(\hat{k}, M, \beta) \} \leq \beta.
\end{aligned}
\end{equation} 
\end{theorem}
It is important to note that the probability above is taken with respect to the $M$ holdout samples $\delta_s^{(1)}, \dots, \delta_s^{(M)}$, with the estimated reachable set $\hat{R}(\theta)$ held fixed. This bound is, virtually, perfectly tight; the bound on the true error of our reachable set is violated exactly a $\beta$ portion of the time \cite{langford2005tutorial}. In Section \ref{applications}, we will demonstrate the efficiency of this method when collecting samples is computationally cheap. 

\textbf{Computation.}
To calculate the binomial tail inversion (cf.~\Cref{def:binomial_tail_inversion}), we use the SciPy Python library to solve an optimization problem with the sequence's cumulative distribution function (CDF) as a constraint, as seen in \eqref{binominversion}. 
Since the binomial tail inversion reduces to inverting a one-dimensional function, it can be computed efficiently using standard optimization routines, such as Quasi-Newton methods. 
Additionally, one can exploit the monotonically decreasing
property of the function $e \mapsto \mathrm{Bin}(k, M, e)$ on the interval $[0, 1]$ and utilize a simple bisection method.

\textbf{Order-wise scaling of binomial tail inversion.}
In general, $\overline{\text{Bin}}(\hat{k}, M, \beta)$ does not have a closed-form solution. However, we can obtain an order-wise scaling to illustrate the dependence on $\hat{k}$, $M$, and $\beta$.
First, suppose that no violations were observed on our holdout dataset ($\hat{k}=0$). 
Then, we have the following bound $\overline{\text{Bin}}(0, M, \beta) \leq \frac{\log(1/\beta)}{M}$ \cite[Corollary 3.4]{langford2005tutorial}. This illustrates a $1/M$ dependence (known as a ``fast-rate'' in statistics) in the number of holdout samples $M$, in addition to a logarithmic dependence, $\log(\frac{1}{\beta} )$, on the failure probability $\beta$. In the general case, when $\hat{k} > 0$, we have the looser scaling
\begin{equation}
    \overline{\text{Bin}}(\hat{k}, M, \beta) \leq \frac{\hat{k}}{M} + O\Bigl(\sqrt{\frac{\log(\frac{1}{\beta})}{M}}\Bigr),
\end{equation}
which is the typical scaling in $M$ predicted by the Central Limit Theorem.

\textbf{Zero violation reachable sets.}
The holdout method, compared to approaches, such as wait-and-judge (which do not require a separate holdout dataset), is not able to guarantee zero violations on the complete dataset. However, we show in \Cref{applications}, that while the wait-and-judge method computes a reachable set estimate with zero dataset violations, the final violation probability, $V(\hat{R}(\theta))$, is substantially more conservative than the holdout method.

\textbf{Marginal probability over training and holdout data.}
\Cref{thm:holdout} provides a bound over the probability of the $M$ holdout samples $\{\delta_s^{(i)}\}_{i=1}^{M}$.
On the other hand, the reachable set $\hat{R}(\theta)$ is calculated as a function of the $N$ training samples $\{\delta^{(i)}\}_{i=1}^{N}$.
It is also possible to obtain a bound which holds over the probability of all $N+M$ data points, which is more comparable to the guarantees provided by scenario optimization.
In fact, the same bound from \Cref{thm:holdout} holds by the tower property:
\begin{align*}
    \mathbf{P}_{(\{\delta^{(i)}\}_{i=1}^{N},\{\delta_s^{(i)}\}_{i=1}^{M})}\{ V(\hat{R}(\theta)) >  \overline{\text{Bin}}(\hat{k}, M, \beta)  \} \leq \beta.
\end{align*}

\section{Applications}
\label{applications}
We present numerical examples to demonstrate that the probabilistic bounds that arise from the holdout method are substantially sharper and require fewer samples than existing methods for data-driven reachability. Specifically, we compare the proposed method to the \textit{wait-and-judge} technique presented in \cite{pmlr-v242-dietrich24a} to illustrate this improvement. Further, we extend scenario-based reachable sets to reachable tubes, an application now amenable to scenario techniques given our improved computational complexity.

\subsection{Reachable Sets}
We calculate reachable sets, using the sum of radial basis functions, posed as a constrained optimization problem  \cite{pmlr-v242-dietrich24a}:
\begin{equation}
\label{eq:rbfs_opt}
\begin{aligned}
& \underset{\mu, \sigma}{\text{minimize}}
& & \sum^m_{i=1} \sigma_i^2\\
& \text{subject to}
& & \sum^m_{i=1} e^{-\frac{1}{2}\frac{(\delta^{(j)}- \mu_i)^2}{\sigma_i^2}}-\gamma \geq 0,& j=1,\dotsc,N,\\
& & & \sigma \in [0,\infty)^m
\end{aligned} 
\end{equation}
Observe that \eqref{eq:rbfs_opt}
corresponds to the program \eqref{eq:volprox}, where 
\begin{equation}
\begin{aligned}
\label{equ:grbf}
g(x,\theta) = \sum^m_{i=1} e^{-\frac{1}{2}\frac{(x - \mu_i)^2}
{\sigma_i^2}} - \gamma
\end{aligned}
\end{equation}
Further, we let Vol$(\hat{R}(\theta)) =\sqrt{\sum_{i=1}^{m} \sigma_i^2}$ be a proxy for the volume of $\hat{R}(\theta)$.
To solve \eqref{eq:rbfs_opt}, we first set the initial centers, $\mu_i$, of the RBFs to be close to optimal using a $k$-means clustering algorithm and choose the initial widths, $\sigma_i$, of our RBFs arbitrarily. We then use the SciPy Python library to solve this optimization problem using Sequential Least Squares Programming. 

We take $\gamma = 0.25$ as the threshold of our RBF and consider 3000 samples. To apply the holdout method, we partition our data set into a training and test set, $N + M = 3000$, and take $\beta = 10^{-9}$. We compute reachable set estimates and $\epsilon$ a-posteriori for various combinations of $N$ and $M$, as seen in Tables \ref{table:duffing-rbf} and \ref{table:quad-rbf}. Additionally, we investigate the difference in volume proxy of our reachable set estimates, to gauge the effect of training set sample complexity. For a comparative baseline, we examine the wait-and-judge technique over $N=3000$ scenarios. When presenting the runtime of the holdout method, this measure includes the calculation of the reachable set over $N$ training samples, test-set generation of $M$ samples, and a-posteriori bound computation. In contrast, runtimes of the wait-and-judge method account for calculation of the reachable set over $N = 3000$ samples and a-posteriori bound computation using support scenarios.

\subsubsection{Duffing Oscillator}
The first example is a reachable set estimation problem for the nonlinear, time-varying system with dynamics: $\ddot{x} =  -\alpha y + x -x^3 + \gamma \cos(\omega t)$, with states $x, y \in \mathbb{R}$ and parameters $\alpha, \gamma, \omega \in \mathbb{R}$. This system is known as the Duffing oscillator, a nonlinear oscillator which exhibits chaotic behavior for certain values of $\alpha, \gamma$ and $\omega$, for instance $\alpha = 0.05, \gamma = 0.4, \omega = 1.3 $. The set of initial states is the interval such that $x(0) \in [0.95, 1.05]$, $y(0) \in [-0.05, 0.05]$, and we take $\mu_{X_0}$ to be the uniform random variable over this interval. The time range is $[t_0, t_1] = [0, 100]$.

\begin{table}[h]\centering
\caption{\normalfont \small Duffing Oscillator: Calculation of $\epsilon$ for the Holdout Method using a Binomial Tail Inversion. 
}
\begin{tabular}{@{}cllccc@{}}
\toprule
Holdout Method:& Training ($N$) & Testing ($M$) & Vol$(\hat{R}(\theta))$ & $\epsilon$  \\ \midrule \midrule
& N = 10 & M = 2990 & 0.57 & 0.776 \\ 
& N = 50 & M = 2950 & 1.27 & 0.147 \\ 
& N = 100 & M = 2900 & 1.44 & 0.089 \\ 
& N = 1000 & M = 2000 & 1.64 & 0.024 \\
& N = 1500 & M = 1500 & 1.54 & 0.018 \\
& N = 2000 & M = 1000 & 1.56 & 0.024 \\ 
& N = 2900 & M = 100 & 1.55 & 0.214 \\ 
& N = 2950 & M = 50 & 1.53 & 0.384 \\
& N = 2990 & M = 10 & 1.53 & 0.874 \\ 
\midrule \midrule
Wait and Judge: & N = 3000 &  & 1.55 & 0.035  \\ 
\bottomrule
\end{tabular}
\label{table:duffing-rbf}
\end{table}
We calculate a reachable set estimate, using two radial basis functions, such that $m = 2$. We apply the holdout method, for combinations of $N$ and $M$, all of which exhibited a runtime of approximately 10-15 sec. When varying the size of the training set, we encounter a significant decrease in volume of our reachable set estimates when the training set has 1000 or less samples. Further, extreme values of $N$, such as $N=10$ or $N=2990$, provide similarly poor probability measures. As seen in Fig. \ref{fig:do-ep}, it is best to balance the size of $N$ and $M$. We observe that $N=1500$, $M=1500$ provides the smallest epsilon, $\epsilon = 0.018$.
In contrast, the wait-and-judge approach exhibited a runtime of approximately 22 min and resulted in $\epsilon = 0.035$, with a volume proxy of Vol$(\hat{R}(\theta)) = 1.55$.   

\begin{figure}[b]
    \centering
    \includegraphics[width=\linewidth]{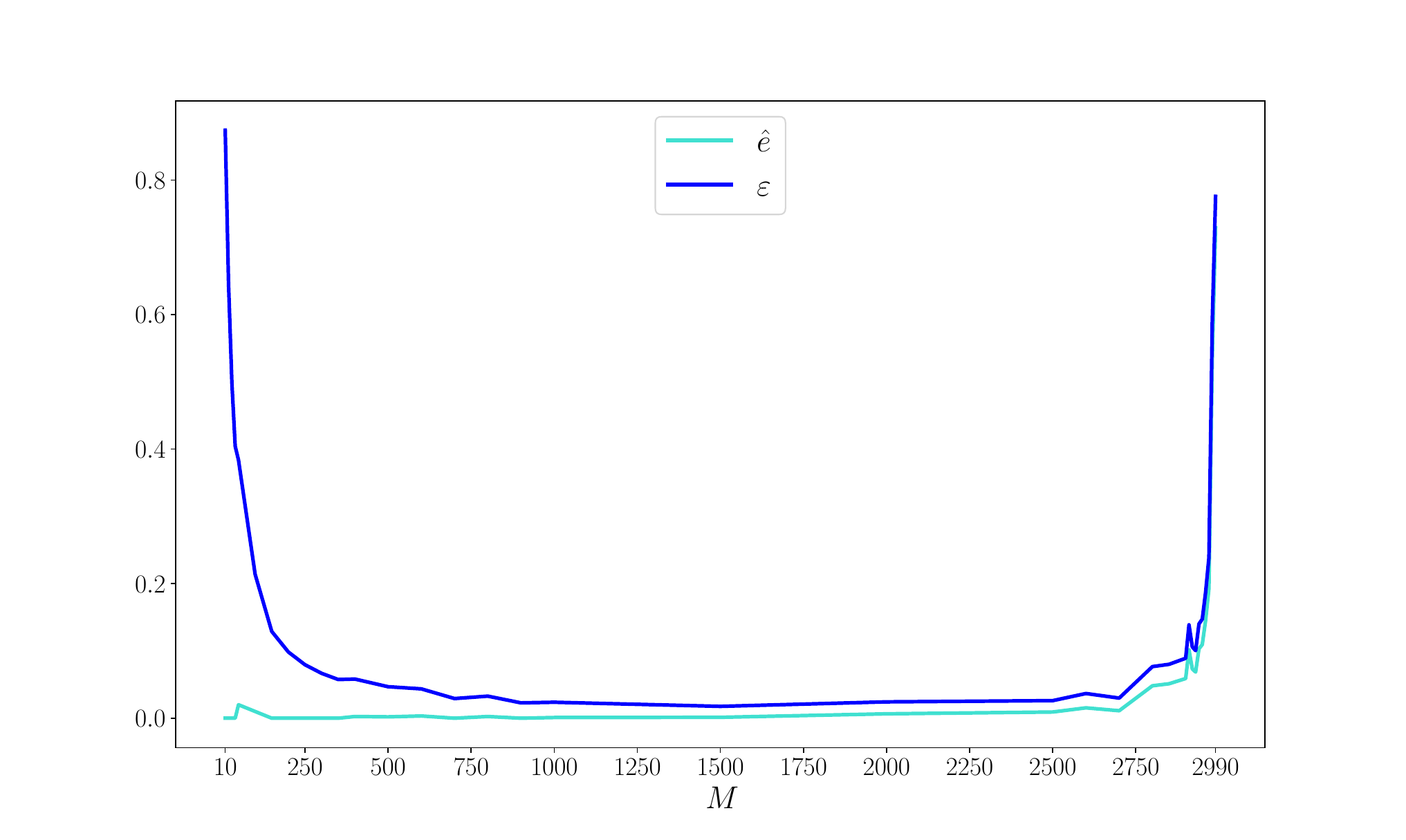}
    \caption{Duffing Oscillator: $\epsilon$ and $\hat{e}$ for various sizes of the holdout dataset. }
    \label{fig:do-ep}
\end{figure}

\subsubsection{Quadrotor}
The next example is for a nonlinear model of a quadrotor used as an example in \cite{10.1145/3302504.3313354, 10.1109/CDC45484.2021.9682860, Bouffard:EECS-2012-241}. The dynamics for this system are
\begin{equation}
    \begin{aligned}
        \ddot{x} &= u_1 K \sin(\theta), \\
        \ddot{h} &= -g + u_1 K \cos(\theta), \\
        \ddot{\theta} &= -d_0 \theta -d_1 \dot{\theta} + n_0 u_2
    \end{aligned}
\end{equation}
where $x$ and $h$ denote the quadrotor's horizontal position and altitude in meters, respectively, and $\theta$ denotes its angular displacement. The system has 6 states, which we take to be $x, h, \theta$, and their first derivatives. The two system inputs $u_1$ and $u_2$ represent the motor thrust and the desired angle, respectively. The parameter values used (following \cite{Bouffard:EECS-2012-241}) are $g = 9.81, K=0.89/1.4, d_0 = 70, d_1=17, n_0=55$. The set of initial states is the interval such that
 \begin{equation}\notag
    \begin{aligned}
        & x(0) \in [-1.7, 1.7],  h(0) \in [0.3, 2.0],  \theta(0) \in [-\pi/12, \pi/12], \\ &\dot{x}(0) \in [-0.8, 0.8],  \dot{h}(0) \in [-1.0, 1.0],   \dot{\theta}(0) \in [-\pi/2, \pi/2],
    \end{aligned}
\end{equation}
the set of inputs is the set of constant functions $u_1(t) = u_1$, $u_2(t) = u_2$ $\forall t \in [t_0, t_1]$, whose values lie in the interval $u_1 \in [-1.5 + g/K, 1.5+g/K], u_2 \in [-\pi/4, \pi/4]$, and we take $\mu_{X_0}$ and $\mu_D$ to be the uniform random variables defined over these intervals. The time range is $[t_0, t_1] = [0,5]$.

We calculate a reachable set estimate using three radial basis functions, such that $m = 3$. We apply the holdout method, for combinations of $N$ and $M$, all of which exhibited a runtime of approximately 35-50 sec. When varying training set size, we observed very similar behavior as that depicted in Fig. \ref{fig:do-ep} for the duffing oscillator. Extreme values of $N$ provide poor probability measures and less than 1000 training samples results in decreased volume of the reachable set estimate. In contrast, the wait-and-judge approach exhibited a runtime of approximately 5.5 hrs and resulted in $\epsilon = 0.051$, with a volume proxy of Vol$(\hat{R}(\theta)) = 27.80$.  

\textbf{Sample Complexity.} In the section above, we utilize equal size datasets to compare the wait-and-judge and holdout method, highlighting the improved accuracy and computational cost of the holdout method. To demonstrate futher improvements in sample complexity, let us examine the duffing oscillator example with sample size $N+M = 2000$ where $N=1000$ and $M=1000$. We apply the holdout method and calculate $\epsilon = 0.0263$, achieving a $1\%$ increase in accuracy with $1000$ fewer samples when compared to the wait-and-judge approach on $3000$ samples.

\begin{table}[H]\centering
\caption{\normalfont \small Quadrotor: Calculation of $\epsilon$ for the Holdout Method using a Binomial Tail Inversion. 
}
\begin{tabular}{@{}cllccc@{}}
\toprule
Holdout Method:& Training ($N$) & Testing ($M$) & Vol$(\hat{R}(\theta))$ & $\epsilon$  \\ \midrule \midrule
& N = 10 & M = 2990 & 11.50 & 0.765  \\
& N = 50 & M = 2950 & 22.97 & 0.280  \\ 
& N = 100 & M = 2900 & 24.29 & 0.120  \\ 
& N = 1000 & M = 2000 & 27.74 & 0.019  \\ 
& N = 1500 & M = 1500 & 27.19 & 0.026 \\ 
& N = 2000 & M = 1000 & 27.85 & 0.031 \\ 
& N = 2900 & M = 100 & 27.81 & 0.187 \\ 
& N = 2950 & M = 50 & 27.30 & 0.384 \\
& N = 2990 & M = 10 & 28.27 & 0.874 \\ 
\midrule \midrule
Wait and Judge: & N = 3000 &  & 27.80 & 0.051  \\ 
\bottomrule
\end{tabular}
\label{table:quad-rbf}
\end{table}

\subsection{Reachable Tubes}
We define a forward estimated reachable tube as $\hat{\mathcal{R}}(t) = \{\Phi(t_i;t_0, x_0, d) : \forall t_i \in [t_0, t], x_0 \in X_0, d \in D\}$ where $[t_0, t]$ is a finite time interval, $X_0 \subseteq \mathbb{R}^{n_x}$ is the set of initial states, $D$ is the set of disturbance signals $d: [t_0, t] \rightarrow \mathbb{R}^{n_d}$, and $\Phi : \mathbb{R} \times X_0 \times D \rightarrow \mathbb{R}^{n_x}$ is the state transition function. $\hat{\mathcal{R}}(t)$ is the set of all states to which the system can transition within a duration of time $[t_0, t]$ for finite time steps $t_i$ from $X_0$ subject to disturbances in $D$. The reachable tube accounts for all time steps in interval $[t_0, t]$, while the reachable sets presented in Section \ref{prelim} only consider one specific time instant. We generalize the nonconvex scenario reachability method presented above and pose the problem of finding a reachable tube directly from data as a chance-constrained optimization problem with regularization over time. In particular, we propose an approach in which we construct a smoothed, finite set of time-varying RBFs: 
\begin{equation}
\notag
\begin{aligned}
&\underset{\mu, \sigma}{\text{minimize}} 
& &  \sum^t_{\tau=0} \sum^m_{i=1} \sigma_i(\tau)^2 + \lambda \| \sigma_{avg} - \sigma_i(\tau)\|^2 \\
& \text{subject to}
& & \sum^t_{\tau=0} \sum^m_{i=1} e^{-\frac{1}{2}\frac{(\delta^{(j)}(\tau)- \mu_i(\tau))^2}{\sigma_i(\tau)^2}}-\gamma \geq 0, j=1,\dotsc,N,\\
& & & \sigma(\tau) \in [0,\infty)^m.
\end{aligned} 
\end{equation}
where $\mu$ is the center of an RBF, $\sigma$ is the width of an RBF, $\lambda$ and $\gamma$ are user-specified parameters, $\delta$ is a scenario, and $\sigma_{avg}$ is the average width over all $\sigma_i$.

To illustrate scenario-based reachable tubes, we investigate a simple linear system: $\dot{x} = Ax$ with 
\begin{align}
A = 
\begin{bmatrix}
-0.7 & -1.0 \\
1.0 & -0.7
\end{bmatrix}.
\end{align}
The set of initial states is the interval such that $x_1(0), x_2(0) \in [1, 1.25]$, and we take $\mu_{X_0}$ to be the uniform random variable over this interval. The time range is $[t_0, t] = [0, 10]$. We calculate the reachable tube estimate using one radial basis function for every time instant. To apply the holdout method, we partition our data set of trajectories into a training and test set, as was done in the previous section, and take $\beta = 10^{-9}$. We consider a boundary violation to be any trajectory that falls outside the reachable tube at any time instant.
\begin{figure}
    \centering
    \includegraphics[width=\linewidth]{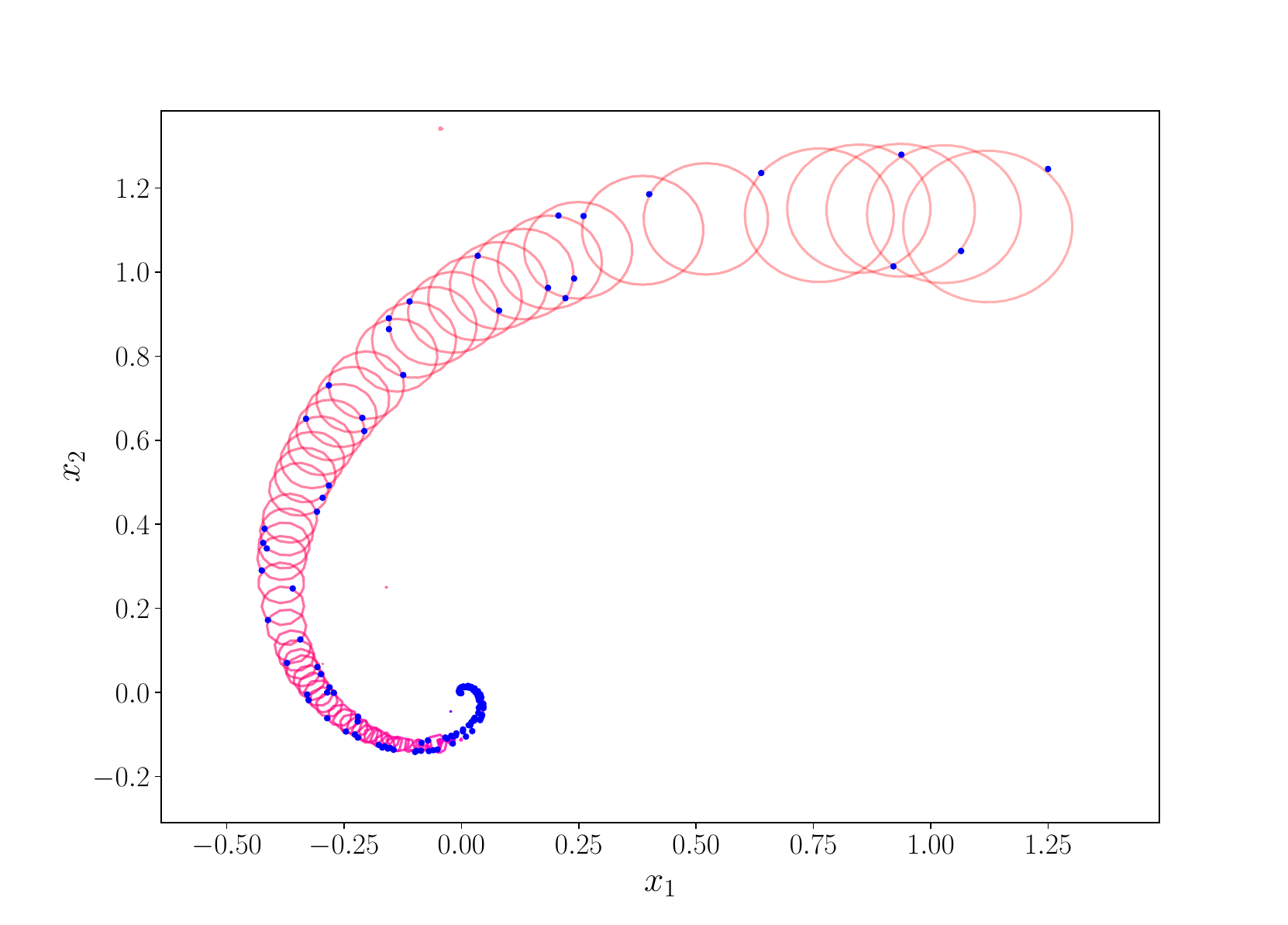}
    \caption{Reachable tube of linear system: The progression of time can be seen as the color of the reachable sets gets darker. The small blue dots indicate the 138 boundary violations that arose from the holdout method. }
    \label{fig:reachtube}
\end{figure}
We present one example, as can be seen in Fig. \ref{fig:reachtube}, given $N=1500$ and $M=1500$. We observed 138 boundary violations, resulting in $\epsilon=0.144$ and a runtime of 4.92 min. Note that by adjusting $\lambda$, we could influence the widths of our RBFs, e.g. increase the size, thereby decreasing the number of violations at later time instances. Further, we do not provide comparisons with the wait-and-judge approach due to its computational cost, as it requires recalculation of the reachable tube upon removal of every individual trajectory.

\section{De-randomization}

In this work, we have advocated for the use of a general purpose, practical, and sharp method for data-driven reachability: the holdout method.
Data-driven reachability offers probabilistic guarantees that are comprised of two ``layers'' of probability (cf.~\Cref{thm:holdout}): the bound on the violation probability, and the bound in probability, with respect to the probability law of the sample generation (in other words, the confidence that the samples have yielded enough information to construct a reachable set estimate for which the violation probability holds).
These types of Probably Approximately Correct (PAC) bounds~\cite{valiant1984pac} are ubiquitous in statistical learning theory and data-driven approaches to probabilistic verification. They are generally recognized as the strongest type of guarantee that can be made while placing minimal assumptions on the random variable verified. 
Nevertheless, it is reasonable to critize two layers of probability on the grounds of interpretability and soundness, compared to conventional control-theoretic guarantees, which provide 100\% certainty given accurate modeling information. To that end, several works have made attempts to ``de-randomize'' data-driven methods, augmenting them with additional information to remove one, or both, of the probabilistic layers~\cite{lecchini2010optimization,esfahani2014performance,pmlr-v155-boffi21a}.

In this section we argue that, while de-randomization aligns with control theory’s pursuit of certainty, it is not practical for the type of data-driven method explored in this paper and similar approaches.
Data-driven verification has two primary advantages:
mild requirements on side information, often requiring none at all, and 
computational efficiency compared to deterministic, sampling-based methods. 
We have found that de-randomized data-driven methods often relinquish both advantages.

\subsection{De-randomization Methods}
\label{sec:derandomization_methods}
The removal of either layer of probability (or both) is carried out by selecting, according to some side information, a suitable enlargement of the data-driven estimator. For simplicity,
we first focus on de-randomizing the inner probability layer, and return to the outer layer later. Specifically, suppose we have found a feasible point $\theta^*$ satisfying the chance constraint $\mathbf{P}\{ g(\theta^*,x) \le 0 \} \ge 1-\epsilon$. To de-randomize such a bound is to establish some $\gamma \ge 0$, as a function of information not available to the original data-driven algorithm, such that $g(\theta^*, x) \le \gamma$ holds almost surely.

A general scheme for how to perform a de-randomization of this type is laid out in~\cite{esfahani2014performance}. The central idea of the scheme is computing a uniform level-set bound function $h(\epsilon)$ that acts as a uniform upper bound (over $\theta$) on the quantile function of the random variable $(\sup_x g(x,\theta) - g(x,\theta))$. As a consequence, it follows that if $\mathbf{P}\{ g(\theta^*,x) \le 0 \} \ge 1-\epsilon$, then 
$\mathbf{P}\{ g(\theta^*,x) \le h(\epsilon) \} =1$.
In our case, the function $g$ contains information about the dynamical system specified in the reachability problem,
so de-randomization requires some amount of new side information. 


A standard assumption 
is to presume knowledge of a global Lipschitz bound on the dynamics and to take the estimator sets to be $p$-norm balls, or ellipsoids. In this case, the magnitude of the enlargement factor from existing work scales exponentially with state dimension \cite[cf. Remark 3.9]{esfahani2014performance}, \cite[cf. Theorem 1]{lecchini2010optimization}.
Interestingly, a less conservative estimate of the reachable set with the same number of queries---and moreover with a non-stochastic guarantee---is possible with a ``sample and cover'' approach. This approach spreads query points uniformly over the initial set and uses a contraction-based approximation to completely cover a region around each query point.
Given this context, 
it is natural to wonder if the situation can be improved with a sharper analysis, or if there is a fundamental limitation with this type of de-randomization in general.

\subsection{Lower Bounds from Zeroth-Order Optimization}
The exponential scaling mentioned previously is inherent to all de-randomization approaches --- ensuring zero violation probability using only samples and a Lipschitz condition.
To demonstrate this,
suppose we are given an $h : B_2^d \mapsto \mathbb{R}$,
where $B_2^d(r) := \{ x \in \mathbb{R}^d \mid \|x\|_2 \leq r\}$ and $B_2^d := B_2^d(1)$,
and we are assured that (a) $h$ is $L$-Lipschitz and
(b) the following probabilistic guarantee holds:
\begin{align*}
    \mathbf{P}_{x \sim \mu(B_2^d)}\{ h(x) > 0 \} \leq \varepsilon, \quad \mu(B_2^d) := \mathrm{Unif}(B_2^d). \label{eq:h_prob_guarantee}
\end{align*}
\begin{lemma}
\label{lemma:lower_bound}
There exists an $h : B_2^d \mapsto \mathbb{R}$ satisfying conditions (a) and (b) above, such that $\max_{x \in B_2^d} h(x) = L \varepsilon^{1/d}$.
\end{lemma}
\begin{proof}
The construction is reminiscent of the Lipschitz bump construction used to prove lower bounds for zeroth-order optimization~\cite[cf. Theorem 1.1.2]{nesterov2018lectures}.
Consider a family of functions $h_\delta(x) := \delta - L \| x \|_2$ for $\delta > 0$.
We now tune $\delta$ so that we ensure
$\mathbf{P}_{x \sim \mu(B_2^d)}\{ h_\delta(x) > 0 \} = \varepsilon$.  Observing that
$$
    \{ x \mid h_\delta(x) > 0 \} = \{ x \mid \delta/ L > \| x \|_2 \} =_{\mathrm{a.e.}} \{ x \mid x \in B_2^d(\delta/L) \},
$$
and using $\mathrm{Vol}(B_2^d(r)) = r^d \mathrm{Vol}(B_2^d(1))$, we have
\begin{align*}
    \mathbf{P}_{x \sim \mu(B_2^d)}\{ h_\delta(x) > 0 \} &= \mathbf{P}_{x \sim \mu(B_2^d)}\{ x \in B_2^d(\delta/L) \} \\
    &= \frac{\mathrm{Vol}(B_2^d(\delta/L))}{\mathrm{Vol}(B_2^d(1))} = (\delta/L)^d.
\end{align*}
Setting $(\delta/L)^d = \varepsilon$ yields $\delta = \delta_\star := L \varepsilon^{1/d}$,
and hence $\max_{x \in B_2^d} h_{\delta_\star}(x) = h_{\delta_\star}(0) = L \varepsilon^{1/d}$.
\end{proof}

The implication of \Cref{lemma:lower_bound} is as follows.
Suppose that the $\varepsilon$ from \Cref{lemma:lower_bound} decreases at a rate of $1/M$, where $M$ is the number of holdout examples, 
one needs at least $(L/\gamma)^d$ samples to ensure a bound $\max_{x \in B_2^d} h(x) \leq \gamma$
in this setting.

Furthermore, from the simple ``sample and cover'' procedure described in \Cref{sec:derandomization_methods}, we see that the same conditions on side information that yield a one-level de-randomization also suffices to de-randomize \emph{both} levels. The basic approach is to remove the possibility of an uninformative sample by selecting the query points non-stochastically, and setting the enlargement factor according to the given Lipschitz bound and the largest distance between any of the query points.
This is equivalent to the well-studied problem of \emph{zeroth-order optimization}, for which well-known lower bounds 
\cite[e.g., Theorem 1.1.2]{nesterov2018lectures}
state that
$(L/\gamma)^d$ queries are needed for any algorithm to obtain the desired $\gamma$-sub-optimality guarantee.\footnote{Nesterov's bound in the original form holds for $\ell_\infty$-Lipschitz (i.e. $| h(x) - h(y) | \leq L \| x - y \|_\infty$) instead of 
$\ell_2$-Lipschitz (i.e., $| h(x) - h(y) | \leq L \|x - y \|_2$), but the proof can be modified for any $\ell_p$-Lipschitz assumption.} Note that this bound precisely matches
the minimum sample complexity prescribed by \Cref{lemma:lower_bound}.

The foregoing demonstrations yield three key takeaways: (a)  De-randomization under standard Lipschitz assumptions necessarily calls for an amount of data that scales \emph{exponentially} with respect to the dimension of the state. (b) The degree of side information required to de-randomize a bound (e.g., Lipschitz constants) is often rich enough to obtain conventional, non-stochastic reachable set over-approximations. (c) There is little room for a ``middle ground'' in de-randomization: the query complexity and side information required to remove one level of probability
is essentially the same as required to remove both. 

In contrast, data-driven methods that yield PAC bounds require almost no side information and are quite sample efficient, as demonstrated in this work. Since the costs incurred by de-randomization are large relative to this baseline, we argue that whenever a sampling-based approach is deemed acceptable, the circumstances where de-randomization is worth the price are rare.





\section{Conclusion}
In this paper, we demonstrate that the holdout method can significantly decrease the sample complexity in finding probabilistically tight reachable sets; this method is highly efficient when collecting scenarios is computationally cheap. Furthermore, we complement our work with a discussion on the necessity of probabilistic reachability bounds within the context of data-driven analysis.

\section{Acknowledgments}

This paper is supported in part by 
the NSF 
project CNS-2111688. The first author was also supported by an NSF Graduate Research Fellowship. 

\bibliography{ref}

\end{document}